\documentclass[a4paper,UKenglish,cleveref, autoref,thm-restate]{lipics-v2021}

\usepackage[ruled]{algorithm}
\usepackage[noend]{algpseudocode}
\usepackage{amsmath}

\hideLIPIcs  

\title{Faster Detours in Undirected Graphs}

\usepackage{mdframed}
\usepackage{framed}
\usepackage[usenames,dvipsnames]{xcolor}

\usepackage{tabularx,tikz,caption,subcaption}
\usetikzlibrary{backgrounds}
\usetikzlibrary{decorations.pathreplacing}

\DeclareMathOperator{\poly}{poly}
\newcommand{\dist}{\textsf{dist}}
\newcommand{\set}[1]{\left\{#1\right\}}
\newcommand{\grp}[1]{\left(#1\right)}

\newcommand{\G}[1]{G_{({#1}, \infty)}}
\newcommand{\Gint}[2]{G_{(#1,#2]}}
\renewcommand{\t}{\text}
\renewcommand{\c}{\mathcal}

\algdef{SE}[SUBALG]{Indent}{EndIndent}{}{\algorithmicend\ }%
\algtext*{Indent}
\algtext*{EndIndent}
\makeatletter
\newcommand{\multiline}[1]{%
  \begin{tabularx}{\dimexpr\linewidth-\ALG@thistlm}[t]{@{}X@{}}
    #1
  \end{tabularx}
}
\makeatother

\author{Shyan Akmal}{MIT EECS and CSAIL, Cambridge MA, USA \and \url{https://www.shyan.akmal.com}}{naysh@mit.edu}{https://orcid.org/0000-0002-7266-2041}{Supported in part by NSF grant CCF-2127597.}
\author{Virginia Vassilevska Williams}{MIT EECS and CSAIL, Cambridge MA, USA \and \url{https://people.csail.mit.edu/virgi/} }{virgi@mit.edu}
{https://orcid.org/0000-0003-4844-2863}{Supported by NSF CCF-2129139.}
\author{Ryan Williams}{MIT EECS and CSAIL, Cambridge MA, USA \and \url{https://people.csail.mit.edu/rrw/} }{rrw@mit.edu}{}{Supported by NSF CCF-1909429.}
\author{Zixuan Xu}{MIT Mathematics, Cambridge MA, USA  }{zixuanxu@mit.edu}{}{}

\authorrunning{S. Akmal, R. Williams, V. V. Williams, and Z. Xu} 

\Copyright{Shyan Akmal, Ryan Williams, Virginia Vassilevska Williams, and Zixuan Xu} 



\category{} 

\relatedversion{} 

\acknowledgements{The first author thanks Nicole Wein for several helpful discussions, and Fedor V. Fomin for answering a question about previous algorithms for \textsf{$k$-Longest Path}.}

\nolinenumbers 

\EventEditors{John Q. Open and Joan R. Access}
\EventNoEds{2}
\EventLongTitle{42nd Conference on Very Important Topics (CVIT 2016)}
\EventShortTitle{CVIT 2016}
\EventAcronym{CVIT}
\EventYear{2016}
\EventDate{December 24--27, 2016}
\EventLocation{Little Whinging, United Kingdom}
\EventLogo{}
\SeriesVolume{42}
\ArticleNo{23}

\begin{document}

\maketitle

\begin{abstract}
    The \textsf{$k$-Detour} problem is a basic path-finding problem: 
    given a graph $G$ on $n$ vertices, with specified nodes $s$ and $t$, and a positive integer $k$, the goal is to determine if $G$ has an $st$-path of length exactly $\dist(s,t) + k$, where $\dist(s,t)$ is the length of a shortest path from $s$ to $t$.
    The \textsf{$k$-Detour} problem is \textsf{NP}-hard when $k$ is part of the input, so researchers have sought efficient parameterized algorithms for this task, running in $f(k)\poly(n)$ time, for $f(\cdot)$ as slow-growing as possible.

    We present faster algorithms for \textsf{$k$-Detour} in undirected graphs, running in $1.853^k\poly(n)$ randomized and $4.082^k\poly(n)$ deterministic time.
    The previous fastest algorithms for this problem took $2.746^k\poly(n)$ randomized and $6.523^k\poly(n)$ deterministic time [Bez{\'a}kov{\'a}-Curticapean-Dell-Fomin, ICALP 2017].
    Our algorithms use the fact that detecting a path of a given length in an undirected graph is easier if we are promised that the path belongs to what we call a ``bipartitioned'' subgraph, where the nodes are split into two parts and the path must satisfy constraints on those parts.
    Previously, this idea was used to obtain the fastest known algorithm for finding paths of length $k$ in undirected graphs [Bj\"{o}rklund-Husfeldt-Kaski-Koivisto, JCSS 2017], intuitively by looking for paths of length $k$ in randomly bipartitioned subgraphs.
    Our algorithms for \textsf{$k$-Detour} stem from a new application of this idea, which does not involve choosing the bipartitioned subgraphs randomly.

    Our work has direct implications for the \textsf{$k$-Longest Detour} problem, another related path-finding problem. In this problem, we are given the same input as in \textsf{$k$-Detour}, but are now tasked with determining if $G$ has an $st$-path of length \emph{at least} $\dist(s,t)+k$. Our results for \textsf{$k$-Detour} imply that we can solve \textsf{$k$-Longest Detour} in $3.432^k\poly(n)$ randomized and $16.661^k\poly(n)$ deterministic time. 
    The previous fastest algorithms for this problem took $7.539^k\poly(n)$ randomized and $42.549^k\poly(n)$ deterministic time [Fomin et al., STACS 2022].
\end{abstract}
\clearpage

\section{Introduction}
\label{sec:intro}

The \textsf{$k$-Path} problem is a well-studied task in computer science: 
\begin{framed} \noindent \textsf{$k$-Path}\\
{\bf Given:} $k \in {\mathbb N}^+$, a graph $G$, nodes $s$ and $t$.\\
{\bf Determine:} Does $G$ contain a simple path of length $k$ from $s$ to $t$?
\end{framed}
For graphs $G$ with $n$ nodes, this problem can be easily solved in $O(kn^k)$ time by enumerating all sequences of $k$ vertices. 
In the 1980s, Monien~\cite{monien-kpath} showed that the \textsf{$k$-Path} problem is actually fixed-parameter tractable (\textsf{FPT}) in $k$, presenting a $k!\poly(n)$ time algorithm solving \textsf{$k$-Path}. 
Since then, significant research has gone into obtaining faster algorithms for \textsf{$k$-Path}, with better dependence on $k$ (see \cite[Table 1]{narrow-sieves} for an overview of the many results in this line of work). 
This research culminated in the work of Koutis and
Williams~\cite{Koutis08,k-path,group-limits}, who showed that \textsf{$k$-Path} can be solved in $2^k\poly(n)$ (randomized) time, and Bj{\"o}rklund, Husfeldt, Kaski, and Koivisto \cite[Section 2]{narrow-sieves}, who proved that in undirected graphs, \textsf{$k$-Path}  can be solved even faster in $1.657^k\poly(n)$ (randomized) time. \emph{Throughout this paper, we assume that algorithms are randomized (and return correct answers with high probability in the stated time bounds), unless otherwise specified.}

The \textsf{$k$-Path} problem is a parameterized version of the $\mathsf{NP}$-complete \textsf{Longest Path} problem, but it is not the only natural parameterization. 
Various other parameterizations of \textsf{$k$-Path} have been proposed and studied, which we consider in the present paper.
\begin{enumerate}
\item {\bf Finding a path of length at least $k$.} Instead of looking for a path of length \emph{exactly} $k$ from $s$ to $t$, one can try to determine the existence of an $st$-path of length \emph{at least} $k$:
\begin{framed}
\noindent \textsf{$k$-Longest Path}\\
\noindent {\bf Given:} $k \in {\mathbb N}^+$, a graph $G$, nodes $s$ and $t$.\\
\noindent {\bf Determine:} Does $G$ contain a simple path of length {\bf \emph{at least}} $k$ from $s$ to $t$?
\end{framed}
Observe that in the \textsf{$k$-Longest Path} problem, the length of a solution path is not necessarily bounded as a function of $k$. 
However, it is known that \textsf{$k$-Longest Path} is also \textsf{FPT}: work of Zehavi~\cite{longest-cycle} and Fomin, Lokshtanov, Panolan, and Saurabh~\cite{longest-path} implies that \textsf{$k$-Longest Path} can be solved in $4^k\poly(n)$ time. 
More recently, Eiben, Koana, and Wahlstr\"{o}m~\cite[Section 6.3]{det-sieving} proved that over undirected graphs, \textsf{$k$-Longest Path} can be solved in $1.657^k\poly(n)$ time, matching the fastest known runtime for \textsf{$k$-Path}.
\item {\bf Finding an $st$-path longer than a polynomial-time guarantee.} Another parameterization for \textsf{$k$-Path} is motivated by the fact that one can find a \emph{shortest path} from $s$ to $t$ in polynomial time. If the shortest path distance $\dist(s,t)$ happens to already be long, then it is actually ``easy'' to find a long path from $s$ to $t$. Therefore, it is natural to consider the parameterized complexity of searching for an $st$-path that is $k$ edges \emph{longer} than the shortest path length from $s$ to $t$. Our work focuses on these so-called ``detour'' variants of the path detection problems discussed above.
\begin{framed}
\textsf{$k$-Detour} (a.k.a. \textsf{$k$-Exact Detour})\\
{\bf Given:} $k \in {\mathbb N}^+$, a graph $G$, nodes $s$ and $t$.\\
{\bf Determine:} Does $G$ contain a simple path of length $\dist(s,t) + k$ from $s$ to $t$?
\end{framed}

Since \textsf{$k$-Path} efficiently reduces to solving a single instance of \textsf{$(k-1)$-Detour},\footnote{Given an instance of \textsf{$k$-Path}, add an edge from $s$ to $t$. Then a solution to \textsf{$(k-1)$-Detour} in this new graph corresponds to a solution to \textsf{$k$-Path} in the original graph.} the \textsf{$k$-Detour} problem is at least as hard as the classical \textsf{$k$-Path}  problem. 

The \textsf{$k$-Detour} problem was introduced by 
Bez{\'a}kov{\'a}, Curticapean, Dell, and Fomin~\cite{detour-original}, who showed that it can be solved by calling polynomially many instances of \textsf{$\ell$-Path}, for path lengths $\ell\le 2k+1$. 
Employing the fastest known \textsf{$k$-Path} algorithms, this implies that \textsf{$k$-Detour} can be solved in $2^{2k}\poly(n) = 4^k\poly(n)$ time in general, and even faster over undirected graphs in $1.657^{2k}\poly(n) \leq 2.746^k\poly(n)$ time.
\end{enumerate}

The two parameterizations above can be combined to produce the following problem:
\begin{framed}
\noindent \textsf{$k$-Longest Detour}\\
{\bf Given:} $k \in {\mathbb N}^+$, a graph $G$, nodes $s$ and $t$.\\
{\bf Determine:} Does $G$ contain a simple path of length {\bf \emph{at least}} $\dist(s,t) + k$ from $s$ to $t$?
\end{framed}

Observe that \textsf{$k$-Longest Detour} is at least as hard as \textsf{$k$-Longest Path}.
Unlike the problems discussed above, \textsf{$k$-Longest Detour} over directed graphs is not known to be \textsf{FPT}: in fact, it remains open whether \textsf{$k$-Longest Detour} is in $\mathsf{P}$ even for the special case of $k=1$! However, in undirected graphs, Fomin, Golovach, Lochet, Sagunov, Simonov, and Saurabh \cite{directed-detours} showed that \textsf{$k$-Longest Detour} can be reduced to solving \textsf{$p$-Detour} for $p\le 2k$, and then solving polynomially many instances of \textsf{$\ell$-Longest Path}, for $\ell \le 3k/2$. Employing the fastest known algorithms for \textsf{$k$-Detour} and \textsf{$k$-Longest Path} as subroutines, this implies that \textsf{$k$-Longest Detour} can be solved over undirected graphs in $7.539^k\poly(n)$ time.

The algorithms for \textsf{$k$-Detour} and \textsf{$k$-Longest Detour} discussed above are significantly slower than the fastest known algorithms for the analogous \textsf{$k$-Path} and \textsf{$k$-Longest Path} problems. 
This motivates the questions: can \textsf{$k$-Detour} be solved as quickly as \textsf{$k$-Path}, and can \textsf{$k$-Longest Detour} be solved as quickly as \textsf{$k$-Longest Path}?
Given the extensive and influential line of work that has gone into finding faster algorithms for \textsf{$k$-Path} and \textsf{$k$-Longest Path}, obtaining faster algorithms for these detour problems as well is an interesting open problem in parameterized complexity and exact algorithms.

\subsection*{Our Results}

The main result of our work is a faster algorithm for \textsf{$k$-Detour} on undirected graphs.

\begin{restatable}{theorem}{detourthm}
    \label{thm:k-detour}
    In undirected graphs, \textsf{$k$-Detour} can be solved in $1.853^k\poly(n)$ time.
\end{restatable} 
\noindent This marks a significant improvement over the previous fastest $2.746^k\poly(n)$ time algorithm for \textsf{$k$-Detour} (and shows, for example, that this problem can be solved in faster than $2^k\poly(n)$ time, which is often a barrier for parameterized problems).
Since the fastest known algorithms for \textsf{$k$-Longest Detour} over undirected graphs have a bottleneck of solving
\textsf{$2k$-Detour}, \Cref{thm:k-detour} implies the following result.

\begin{restatable}{theorem}{longestdetour}
\label{thm:k-longest-detour}
In undirected graphs, \textsf{$k$-Longest Detour}  can be solved in $3.432^k\poly(n)$ time.
\end{restatable}
\noindent Again, this is a significant improvement over the previous fastest algorithm for \textsf{$k$-Longest Detour} on undirected graphs, which ran in $7.539^k\poly(n)$ time.

Our algorithm for Theorem~\ref{thm:k-detour} applies the fact that \textsf{$k$-Path} is easier to solve on undirected graphs which have a prescribed vertex partition into two sets, where we constrain the path to contain a particular number of nodes from one set, and a particular number of edges whose vertices are in the other set. 
Formally, we consider the \textsf{$(\ell, k_1, \ell_2)$-Bipartitioned Path} problem: given a graph $G$ on $n$ nodes, whose vertices are partitioned into parts $V_1$ and $V_2$, with distinguished vertices $s$ and $t$, the goal is to determine if $G$ contains a simple path from $s$ to $t$ of length $\ell$, which uses exactly $k_1$ vertices from $V_1$, and exactly $\ell_2$ edges whose endpoints are both in $V_2$.
A careful application of the following result from \cite{narrow-sieves} is the main source of the speed-up in our algorithm for \textsf{$k$-Detour}.

\begin{restatable}[{\cite[Section 2]{narrow-sieves}}]{lemma}{bipartite}
    \label{lem:near-bipartite} Let $\ell, k_1, \ell_2$ be nonnegative integers satisfying the inequality $\ell+1 \geq k_1 + 2 \ell_2$. Then over undirected graphs, the \textsf{$(\ell, k_1, \ell_2)$-Bipartitioned Path} problem can be solved in $2^{k_1+\ell_2}\poly(n)$ time.
\end{restatable}

Although this ``Bipartitioned Path'' problem may appear esoteric at first, Lemma~\ref{lem:near-bipartite} plays a crucial role in obtaining the fastest known algorithm for \textsf{$k$-Path} in undirected graphs \cite{narrow-sieves}, and an analogue of \Cref{lem:near-bipartite} for paths of length at least $k$ is the basis for the fastest known algorithm for \textsf{$k$-Longest Path} in undirected graphs \cite{det-sieving}.
For completeness, we include a proof of \Cref{lem:near-bipartite} in \Cref{sec:near-bipartite}.
In \Cref{sec:overview}, we provide an intuitive overview of how  \Cref{lem:near-bipartite} helps us obtain our algorithm for \textsf{$k$-Detour}.

The fastest known algorithms for the path and detour problems discussed above all use randomness.
Researchers are also interested in obtaining fast \emph{deterministic} algorithms for these problems.
We note that a simplified version of our algorithm for \textsf{$k$-Detour} implies faster deterministic algorithms for these detour problems over undirected graphs.

\begin{restatable}{theorem}{detdetour}
    \label{det-detour}
    The \textsf{$k$-Detour} problem can be solved over undirected graphs by a deterministic algorithm in $4.082^k\poly(n)$ time.
\end{restatable}
\noindent Prior to this work, the fastest known deterministic algorithm for \textsf{$k$-Detour} on undirected graphs ran in $6.523^k\poly(n)$ time \cite{detour-original}.

\begin{restatable}{theorem}{detlongestdetour}
    \label{det-longestdetour}
    The \textsf{$k$-Longest Detour} problem can be solved over undirected graphs by a deterministic algorithm in $16.661^k\poly(n)$ time.
\end{restatable}
\noindent Prior to this work, the fastest known deterministic algorithm for \textsf{$k$-Longest Detour} on undirected graphs ran in $42.549^k\poly(n)$ time \cite{directed-detours}.

In summary, we obtain new randomized and deterministic algorithms for \textsf{$k$-Detour} and \textsf{$k$-Longest Detour} over undirected graphs, whose runtimes present significant advances over what was previously known for these problems.

\subsection{Organization}
The remainder of the main body of this paper presents our new algorithms \textsf{$k$-Detour}.
In  \Cref{sec:related-work} we include a thorough discussion of additional related work, and in \Cref{sec:near-bipartite} we include a proof of \Cref{lem:near-bipartite}.

In \Cref{sec:prelim}, we introduce relevant notation, assumptions, and definitions concerning graphs. 
In \Cref{sec:overview}, we provide an overview of our algorithm for \textsf{$k$-Detour}.
In \Cref{sec:detour}, we present the details of our algorithm, and prove its correctness.
The  runtime analysis for our algorithm (and thus the formal proofs of \Cref{thm:k-detour,thm:k-longest-detour,det-detour,det-longestdetour}, given correctness of our algorithm) are presented in \Cref{subsec:app}.
In \Cref{sec:conclusion}, we highlight some open problems.

\section{Preliminaries}
\label{sec:prelim}

Given positive integers $a$ and $b$, we let $[a] = \set{1, 2, \dots, a}$, and $[a,b] = \set{a, a+1, \dots, b}$.
Given an integer $a$ and a set of integers $S$, we define
$a + S = \set{a + s\mid s\in S}.$

Throughout, we let $G$ denote the input graph.
We assume that $G$ is undirected, has vertex set $V$ with $|V|=n$, and, without loss of generality, is connected.\footnote{If $G$ were not connected, we could replace it with the connected component containing $s$, and solve the detour problems on this smaller graph instead.}
Throughout, we let $s$ and $t$ denote the two distinguished vertices that come as part of the input to the \textsf{$k$-Detour} problem.
Given a vertex $u$, we let $d(u) = \dist(s,u)$ denote its distance from $s$. 
This distance is well-defined, since $G$ is connected.
Given a path $P$ containing vertices $u$ and $v$, we let $P[u,v]$ denote the subpath from $u$ to $v$ on $P$.

Given an edge $e = (u,v)$ from $u$ to $v$, we say $e$ is {\bf \emph{forward}} if $d(v) = d(u) + 1$, {\bf \emph{backwards}} if $d(v) = d(u)-1$, and {\bf \emph{stable}} if $d(v) = d(u)$. 
In an undirected graph, by triangle inequality and symmetry of distance, adjacent vertices $u$ and $v$ have $|d(u)-d(v)|\le 1$, so every edge in a path falls into one of these three categories.

Given two vertices $u,v\in V$, let $\Gint{u}{v}$ denote the induced subgraph of $G$ on the set $\set{u}\cup\set{w\in V\mid d(u) < d(w)\le d(v)}$. Let $\G{u}$ denote the induced subgraph of $G$ on the set $\set{u}\cup \set{w\in V\mid d(u) < d(w)}$. Note that for every $u$ and $v$, the subgraphs $\Gint{u}{v}$ and $\G{v}$ overlap at vertex $v$, but are disjoint otherwise.

\section{Technical Overview}
\label{sec:overview}
In this section, we provide an overview of how our \textsf{$k$-Detour} algorithm works.
Our starting point is the algorithm for this problem presented in \cite[Section 4]{detour-original}, which we review in \Cref{subsec:orig-detour}.
Then in \Cref{subsec:near-bipartite} we review how the algorithm from \Cref{lem:near-bipartite} for \textsf{$(\ell,k_1,\ell_2)$-Bipartitioned Path} has previously been used to obtain the fastest known algorithm for \textsf{$k$-Path} in undirected graphs. 
With this context established, in \Cref{subsec:improvement} we outline how we combine the techniques from \Cref{subsec:orig-detour,subsec:near-bipartite} with new ideas to prove \Cref{thm:k-detour}.

\subsection{Previous Detour Algorithm}
\label{subsec:orig-detour}

The previous algorithm for \textsf{$k$-Detour} from \cite[Section 4]{detour-original} performs dynamic programming over nodes in the graph, starting from $t$ and moving to vertices closer to $s$. 
In the dynamic program, for each vertex $x$ with $d(x)\le d(t)$, we compute all offsets $r\le k$ such that there is a path of length $d(t)-d(x) + r$ from $x$ to $t$ in the subgraph $\G{x}$.
Determining this set of offsets for $x=s$ solves the \textsf{$k$-Detour} problem, since $\G{s} = G$.

If $d(t)-d(x)\le k$, we can find all such offsets just by solving \textsf{$\ell$-Path} for $\ell\le 2k$.

So, suppose we are given a vertex $x$ with $d(t)-d(x) \ge k+1$ and an offset $r\le k$, and wish to determine if there is a path of length $d(t)-d(x) + r$ from $x$ to $t$ in $\G{x}$.
If there is such a path $P$, then \cite{detour-original} argues that $P$ can always be split in as depicted in \Cref{subfig:old-split}: for some vertex $y$ with $d(y) > d(x)$, we can decompose $P$ into two subpaths: 
\begin{enumerate}
\item a subpath $A$ from $x$ to $y$ of length at most $2k+1$, such that all internal vertices $v$ in $A$ satisfy $d(x) < d(v) < d(y)$, and 
\item a subpath $B$ from $y$ to $t$ in $\G{y}$ of length at most $d(t)-d(y)+k$.
\end{enumerate}

\begin{figure}
    \centering
    \begin{subfigure}[t]{0.48\textwidth}
    \centering
        \begin{tikzpicture}[scale=0.7,
    vtx/.style={circle,draw,fill=black,inner sep=0pt,minimum width=3pt},
    edge/.style={-stealth, thick}
    ]

    \def\top{4.5cm};
    \def\bottom{0.5cm};
    \def\half{{(0.5*(\top+\bottom))}};
    
    \def\start{0.25cm};
    \def\layershift{1.3cm};

    \foreach \x in {0,...,7} {
        \draw[gray,dashed] (\start + \layershift*\x,\bottom) -- (\start + \layershift*\x,\top);
    }

    \node[vtx] (1) at (\start,\half) {};
    \node[scale=1.2] at (\start + 0.4cm, 2.1cm) {$x$};

    \node[vtx] (2) at (\start + \layershift, 4) {};
    \node[vtx] (4) at (\start + \layershift, 3.25) {};
    \node[vtx] (5) at (\start + \layershift, 2.5) {};
    \node[vtx] (6) at (\start + \layershift, 1.75) {};

    \node[vtx] (3) at (\start + 2*\layershift, 4) {};
    \node[vtx] (7) at (\start + 2*\layershift, 1) {};
    \node[vtx] (8) at (\start + 2*\layershift, 1.75) {};
    \node[vtx] (9) at (\start + 2*\layershift, 2.5) {};

    \node[vtx] (10) at (\start + 3*\layershift, 1.75) {};
    \node[vtx] (11) at (\start + 3*\layershift, 2.5) {};

    \node[vtx,color=blue] (12) at (\start + 4*\layershift, 2.5) {};
    \node[scale=1.2,color=blue] at (\start + 4*\layershift + 0.4cm, 2.1) {$y$};

    \node[vtx] (13) at (\start + 5*\layershift, 3.25) {};
    \node[vtx] (14) at (\start + 5*\layershift, 2.5) {};
    \node[vtx] (15) at (\start + 5*\layershift, 1.75) {};

    \node[vtx] (16) at (\start + 6*\layershift, 1.75) {};
    \node[vtx] (20) at (\start + 6*\layershift, 2.5) {};
    \node[scale=1.2] at (\start + 6*\layershift - 0.4cm, 2.8) {$t$};

    \node[vtx] (17) at (\start + 7*\layershift, 1.75) {};
    \node[vtx] (18) at (\start + 7*\layershift, 2.5) {};
    \node[vtx] (19) at (\start + 7*\layershift, 3.25) {};

    \draw[edge]   (1) -- (2);
    \draw[edge]   (2) -- (3);
    \draw[edge]   (3) -- (4);
    \draw[edge]   (4) -- (5);
    \draw[edge]   (5) -- (6);
    \draw[edge]   (6) -- (7);
    \draw[edge]   (7) -- (8);
    \draw[edge]   (8) -- (9);
    \draw[edge]  (9) -- (10);
    \draw[edge] (10) -- (11);
    \draw[edge] (11) -- (12);
    \draw[edge] (12) -- (13);
    \draw[edge] (13) -- (14);
    \draw[edge] (14) -- (15);
    \draw[edge] (15) -- (16);
    \draw[edge] (16) -- (17);
    \draw[edge] (17) -- (18);
    \draw[edge] (18) -- (19);
    \draw[edge] (19) -- (20);
   
    \begin{scope}[on background layer]   
      \filldraw[Yellow!50, line width=15pt,rounded corners=4pt] 
      (1.center) -- (\start + \layershift, \top) -- (\start + 3*\layershift, \top) -- (12.center) -- (\start + 3*\layershift, \bottom) -- (\start + \layershift, \bottom) -- cycle;
    \end{scope}
    \begin{scope}[on background layer]   
      \filldraw[WildStrawberry!15, line width=15pt,rounded corners=4pt] (12.center) --  (\start + 5*\layershift, \top) --  (\start + 7*\layershift, \top) -- (\start + 7*\layershift, \bottom) --  (\start + 5*\layershift, \bottom) --  cycle;
    \end{scope}

    \draw [decorate,decoration={brace,amplitude=5pt,raise=4ex}] (\start + \layershift - 0.1cm, 4) -- (\start + 4*\layershift + 0.1cm, 4);
    \node[scale=1] at (\start + 2.5*\layershift, 5.6) {$\le k+1$};
    
    \end{tikzpicture}
    \caption{\textbf{Previous Algorithms}: A subpath $P$ from $x$ to $t$ in $\G{x}$ of a solution path can always be split at a vertex $y$ with $d(y)\le d(x)+k+1$, such that $d(u)\neq d(y)$ for all vertices $u\neq y$ in $P$.}
    \label{subfig:old-split}
    \end{subfigure}
    \hfill
    \begin{subfigure}[t]{0.48\textwidth}
    \centering
    \begin{tikzpicture}[scale=0.7,
    vtx/.style={circle,draw,fill=black,inner sep=0pt,minimum width=3pt},
    edge/.style={-stealth, thick}
    ]

    \def\top{4.5cm};
    \def\bottom{0.5cm};
    \def\half{{(0.5*(\top+\bottom))}};
    
    \def\start{0.25cm};
    \def\layershift{1.3cm};

    \foreach \x in {0,...,7} {
        \draw[gray,dashed] (\start + \layershift*\x,\bottom) -- (\start + \layershift*\x,\top);
    }

    \node[vtx] (1) at (\start,\half) {};
    \node[scale=1.2] at (\start + 0.4cm, 2.1cm) {$x$};

    \node[vtx] (2) at (\start + \layershift, 4) {};
    \node[vtx] (4) at (\start + \layershift, 3.25) {};
    \node[vtx] (5) at (\start + \layershift, 2.5) {};
    \node[vtx] (6) at (\start + \layershift, 1.75) {};

    \node[vtx] (3) at (\start + 2*\layershift, 4) {};
    \node[vtx] (7) at (\start + 2*\layershift, 1) {};
    \node[vtx] (8) at (\start + 2*\layershift, 1.75) {};
    \node[vtx,color=blue] (9) at (\start + 2*\layershift, 2.5) {};
    \node[scale=1.2,color=blue] at (\start + 2*\layershift - 0.4cm, 2.8) {$y$};

    \node[vtx] (10) at (\start + 3*\layershift, 1.75) {};
    \node[vtx] (11) at (\start + 3*\layershift, 2.5) {};

    \node[vtx] (12) at (\start + 4*\layershift, 2.5) {};

    \node[vtx] (13) at (\start + 5*\layershift, 3.25) {};
    \node[vtx] (14) at (\start + 5*\layershift, 2.5) {};
    \node[vtx] (15) at (\start + 5*\layershift, 1.75) {};

    \node[vtx] (16) at (\start + 6*\layershift, 1.75) {};
    \node[vtx] (20) at (\start + 6*\layershift, 2.5) {};
    \node[scale=1.2] at (\start + 6*\layershift - 0.4cm, 2.8) {$t$};

    \node[vtx] (17) at (\start + 7*\layershift, 1.75) {};
    \node[vtx] (18) at (\start + 7*\layershift, 2.5) {};
    \node[vtx] (19) at (\start + 7*\layershift, 3.25) {};

    \draw[edge]   (1) -- (2);
    \draw[edge]   (2) -- (3);
    \draw[edge]   (3) -- (4);
    \draw[edge]   (4) -- (5);
    \draw[edge]   (5) -- (6);
    \draw[edge]   (6) -- (7);
    \draw[edge]   (7) -- (8);
    \draw[edge]   (8) -- (9);
    \draw[edge]  (9) -- (10);
    \draw[edge] (10) -- (11);
    \draw[edge] (11) -- (12);
    \draw[edge] (12) -- (13);
    \draw[edge] (13) -- (14);
    \draw[edge] (14) -- (15);
    \draw[edge] (15) -- (16);
    \draw[edge] (16) -- (17);
    \draw[edge] (17) -- (18);
    \draw[edge] (18) -- (19);
    \draw[edge] (19) -- (20);
   
    \begin{scope}[on background layer]   
      \filldraw[Yellow!50, line width=15pt,rounded corners=4pt] 
      (1.center) -- (\start + \layershift, \top) -- (\start + 2*\layershift, \top) -- (\start + 2*\layershift, \bottom) -- (\start + \layershift, \bottom) -- cycle;
    \end{scope}
    \begin{scope}[on background layer]   
      \filldraw[WildStrawberry!15, line width=15pt,rounded corners=4pt] (9.center) --  (\start + 3*\layershift, \top) --  (\start + 7*\layershift, \top) -- (\start + 7*\layershift, \bottom) --  (\start + 3*\layershift, \bottom) --  cycle;
    \end{scope}

    \node[scale=1.2] at (\start + 2*\layershift - 0.65cm, 0.6) {$G_{(x,y]}$};
    \node[scale=1.2] at (\start + 5*\layershift , 0.6) {$G_{(y,\infty)}$};

    \draw [decorate,decoration={brace,amplitude=5pt,raise=4ex}] (\start + \layershift - 0.1cm, 4) -- (\start + 2*\layershift + 0.1cm, 4);
    \node[scale=1] at (\start + 1.5*\layershift, 5.6) {$\le k/2+1$};
    
    \end{tikzpicture}
    \caption{\textbf{Our Algorithm}: In undirected graphs, a subpath $P$ from $x$ to $t$ in $\G{x}$ of a solution path  can be split at a vertex $y$ with $d(y)\le d(x) + k/2 + 1$, such that $P[x,y]$ is in $\Gint{x}{y}$ and $P[y,t]$ is in $\G{y}$.}
    \label{subfig:new-split}
    \end{subfigure}
    \caption{To solve \textsf{$k$-Detour}, we split subpaths $P$ of candidate solutions at a vertex $y$ satisfying certain nice properties.
    We obtain a speed-up by getting better upper bounds on $d(y)$ in \Cref{subfig:new-split} than previous work did in \Cref{subfig:old-split}, by allowing $P[x,y]$ to have internal vertices $u$ with $d(u) = d(y)$.}
    \label{fig:path-split-comparison}
\end{figure} 

We can always split a path $P$ in this manner because $P$ has length at most $d(t)-d(x)+k$, 
so at most $k$ edges in $P$ are not forward edges.
Intuitively, as we follow the vertices along the path $P$, the distance of the current vertex from $s$ can decrease or stay the same at most $k$ times, and so $P$ cannot contain too many vertices which are the same distance from $s$. 
This allows one to argue that there is a vertex $y$ with $d(y)\le d(x)+k+1$ such that all internal vertices $v$ of the subpath $P[x,y]$ have $d(x) < d(v) < d(y)$. Since $d(y)\le d(x)+k+1$ and $P$ has length at most $d(t)-d(x)+k$, it turns out that $P[x,y]$ has length at most $2k+1$.

Note that since $\G{y}$ only contains vertices $v$ with $d(v) \geq y$, the paths $A$ and $B$ must be disjoint. We can find the length of $A$ using an algorithm for \textsf{$(2k+1)$-Path}, and the length of $B$ will have already been computed in our dynamic program (since $y$ is further from $s$ than $x$). 
So, by trying out all possible $y$, finding the possible lengths for subpaths $A$ and $B$, and then adding up these lengths, we can get all possible lengths for $P$ in the dynamic program, and solve \textsf{$k$-Detour}.

\subsection{Previous Path Algorithm}
\label{subsec:near-bipartite}


The fastest known algorithm for \textsf{$k$-Path} in undirected graphs goes through the \textsf{$(k,k_1,\ell_2)$-Bipartitioned Path} problem.
Recall that in this problem, we are given a bipartition $V_1\sqcup V_2$ of the vertices in the graph, and want to find a path of length $k$ from $s$ to $t$, which uses $k_1$ vertices in $V_1$ and $\ell_2$ edges with both endpoints in $V_2$.
The authors of \cite{narrow-sieves} showed that \textsf{$(k,k_1,\ell_2)$-Bipartitioned Path} can be solved in $2^{k_1 + \ell_2}\poly(n)$ time over undirected graphs.

Why does this imply a faster algorithm for \textsf{$k$-Path} in undirected graphs?
Well, suppose the input graph contains a path $P$ of length $k$ from $s$ to $t$.
Consider a uniform random bipartition of the vertices of the graph into parts $V_1$ and $V_2$.
We expect $(k+1)/2$ vertices of $P$ to be in $V_1$, and $k/4$ edges of $P$ to have both endpoints in $V_2$.
In fact, this holds with constant probability, so we can solve \textsf{$k$-Path} by solving \textsf{$(k, (k+1)/2, k/4)$-Bipartitioned Path} in the randomly partitioned graph. 
By \Cref{lem:near-bipartite} this yields a $2^{3k/4}\poly(n) \approx 1.682^k\poly(n)$ time algorithm for \textsf{$k$-Path}. 
We can obtain a faster algorithm using the following modification: take several uniform random bipartitions of the graph, and solve \textsf{$(k, k_1, \ell_2)$-Bipartitioned Path} separately for each bipartition, for
$k_1 + \ell_2 \le 3(1-\varepsilon)k/4,$ where $\varepsilon > 0$ is some constant.
The number of bipartitions used is some function of $k$ and $\varepsilon$, set so that with high probability at least one of the partitions $V_1\sqcup V_2$ has the property that the total number of vertices of $P$ in $V_1$ and number of edges of $P$ with both endpoints in $V_2$ is at most $3(1-\varepsilon)k/4$.
Setting the parameter $\varepsilon$ optimally yields a $1.657^k\poly(n)$ time algorithm for \textsf{$k$-Path} \cite{narrow-sieves}.

\subsection{Our Improvement}
\label{subsec:improvement}

As in the previous approach outlined in \Cref{subsec:orig-detour}, our algorithm for \textsf{$k$-Detour} performs dynamic programming over vertices in the graph, starting 
at $t$, and then moving to vertices closer to $s$.
For each vertex $x$ with $d(x)\le d(t)$, we compute all offsets $r\le k$ such that there is a path of length $d(t)-d(x) + r$ from $x$ to $t$ in the subgraph $\G{x}$. 
Obtaining this information for $x=s$ and $r=k$ solves the \textsf{$k$-Detour} problem.

Given a vertex $x$ and offset $r\le k$, we wish to determine if $G$ contains a path of length $d(t)-d(x) + r$ from $x$ to $t$ in $\G{x}$.
Suppose there is such a path $P$.
If $d(t)-d(x)$ is small enough, it turns out we can find $P$ by solving \textsf{$p$-Path} for small values of $p$.
So, for the purpose of this overview, suppose that $d(t)-d(x)$ is sufficiently large.
In this case, as outlined in \Cref{subsec:orig-detour}, previous work showed that $P$ can be split into two subpaths $A$ and $B$ contained in disjoint subgraphs, such that $A$ has length at most $2k+1$.
This splitting argument 
holds even for directed graphs.
Our first improvement comes from the observation that in undirected graphs, we can decompose the path $P$ with a smaller prefix: as depicted in \Cref{subfig:new-split}, there must exist a vertex $y$ with $d(y) > d(x)$, such that $P$ splits into a subpath $A$ from $x$ to $y$ in $\Gint{x}{y}$ of length at most $3k/2+1$, and a path $B$ from $y$ to $t$ in $\G{y}$ of length at most $d(t)-d(y)+k$.
We can find the length of $A$ by solving \textsf{$(3k/2+1)$-Path}, and the length of $B$ will already have been computed by dynamic programming, since $d(y) > d(x)$.

This split is possible because any consecutive vertices $u$ and $v$ in $P$ have $|d(u)-d(v)|\le 1$ (this is true for undirected graphs, but is not true in general for directed graphs). 
Since $P$ has length at most $d(t)-d(x)+k$, it turns out that $P$ has at most $k/2$ backwards edges.
This lets us argue that there exists a vertex $y$ with $d(y)\le d(x)+k/2+1$ such that $P[x,y]$ is contained in $\Gint{x}{y}$ and $P[y,t]$ is contained in $\G{y}$.
Finally, $A = P[x,y]$ should have length at most $k$ more than $d(y)-d(x)$, which means it has length at most $3k/2+1$.

This simple modification already yields a faster algorithm\footnote{In fact, this observation already yields the fast deterministic algorithms for \Cref{det-detour,det-longestdetour}.} for \textsf{$k$-Detour}.
We get further improvements by performing casework on the number of stable edges in $P$ (recall that an edge $(u,v)$ is stable if both its endpoints have the same distance $d(u)=d(v)$ from $s$).

First, suppose $P$ has at least $m$ stable edges, for some parameter $m$. 
Since $P$ has length at most $d(t)-d(x)+k$, we can argue that $P$ has at most $(k-m)/2$ backwards edges.
With this better upper bound on the number of backwards edges, we can improve the splitting argument and show that $P$ decomposes into subpaths $A$ and $B$, such that the length of $A$ is at most $(3k-m)/2$, and the length of $B$ was already computed by our dynamic program.
It then suffices to solve \textsf{$(3k-m)/2$-Path}, which yields a speed-up  whenever $m\ge \Omega(k)$.

Otherwise, $P$ has at most $m$ stable edges. In this case, we consider the bipartition $V_1\sqcup V_2$ of the vertex set, where $V_1$ has all vertices at an odd distance from $s$, and $V_2$ has all vertices with even distance from $s$. 
Since $G$ is undirected, consecutive vertices on the path $P$ differ in their distance from $s$ by at most one.
In particular, all forward and backward edges in $P$ cross between the parts $V_1$ and $V_2$.
Only the stable edges can contribute to edges with both endpoints in $V_2$.
Since we assumed that the number of stable edges is small, it turns out we can find the length of the subpath $A$ of $P$ by solving \textsf{$(\ell,k_1,\ell_2)$-Bipartitioned Path} with respect to the given bipartition, for some $\ell_2$ which is very small.
In particular, this approach computes the length of $A$ faster than naively solving \textsf{$(3k-m)/2$-Path}.
By setting an appropriate threshold for $m$, we can minimize the runtimes of the algorithm in both of the above cases, and establish \Cref{thm:k-detour}. 

So in summary, our faster algorithms come from two main sources of improvement: using the structure of shortest paths in undirected graphs to get a better ``path-splitting'' argument in the dynamic program from \textsf{$k$-Detour}, and cleverly applying the fast algorithm from \Cref{lem:near-bipartite} for \textsf{$(\ell,k_1,\ell_2)$-Bipartitioned Path} with carefully chosen bipartitions.

We note that our application of \textsf{$(\ell,k_1,\ell_2)$-Bipartitioned Path} is qualitatively different from its uses in previous work.
As discussed in \Cref{subsec:orig-detour}, previous algorithms for \textsf{$k$-Detour} work by solving instances of \textsf{$k$-Path}, and
as described in \Cref{subsec:near-bipartite}, the fastest algorithms for \textsf{$k$-Path} on undirected graphs work by reduction to various instances of \textsf{$(\ell,k_1,\ell_2)$-Bipartitioned Path}.
Thus, previous algorithms for \textsf{$k$-Detour} on undirected graphs  \emph{implicitly} rely on the fast algorithm for \textsf{$(\ell,k_1,\ell_2)$-Bipartitioned Path}, applied to random bipartitions of the input graph.
We obtain a faster algorithm for \textsf{$k$-Detour} arguing that in certain cases, we can ``beat randomness,'' by constructing bipartitions which leverage structural information about the graph (namely, whether the shortest path distance from $s$ to a given vertex is even or odd). 

\section{Detour Algorithm}
\label{sec:detour}

In this section, we present \Cref{alg:detour}, our new algorithm for the \textsf{$k$-Detour} problem. 
As mentioned in the previous section, our algorithm behaves differently depending on the number of stable edges that a potential solution path contains. 
In particular, the algorithm depends on a parameter $\alpha\in (0,1)$, which determines the threshold for what counts as ``many'' stable edges. 
Later, we will set $\alpha$ to optimize the runtime of \Cref{alg:detour}.
Certain lines of \Cref{alg:detour} have comments indicating case numbers, which are explained in \cref{subsec:correctness}.

\begin{algorithm}[t!]
\caption{Our algorithm for solving \textsf{$k$-Detour} in undirected graphs.}
\label{alg:detour}

    \begin{algorithmic}[1]  

    \Require An undirected graph $G$ with distinguished vertices $s$ and $t$, and a parameter $\alpha\in (0,1)$.
    
    \hfill

    \State Initialize $V_1\gets \{x\in V\mid d(x) \equiv 1 \mod 2\}$, $V_2\gets \{x\in V \mid d(x)\equiv 0 \mod 2\}$. \label{step:partition}
    
    \State For each vertex $x$ in the graph with $d(x)\le d(t)$, initialize $L(x) \leftarrow \emptyset$.

    \Statex 
        \begin{quote}$\rhd$
        \textit{
        $L(x)$ will be the set of lengths $\ell\in [d(t)-d(x), d(t)-d(x)+k]$ such that there is a path of length $\ell$ from $x$ to $t$ in $\G{x}$.
        }
        \end{quote}
        \hfill 

    \State For each vertex $x$ with $d(x)\in [d(t)-(1-\alpha)k/2, d(t)]$, set $L(x)$ to be the set of all positive integers $\ell\le (3-\alpha)k/2$ such that there is a path of length $\ell$ from $x$ to $t$ in $\G{x}$.
    \label{step:base}

    \Statex
        \begin{quote}$\rhd$
        \textit{
        Base case: we compute $L(x)$ for the vertices $x$ which are furthest from $s$. 
        }
        \end{quote}
        \hfill

    \State For each $d$ from $d(t)-(1-\alpha)k/2-1$ down to $0$, for each vertex $x$ with $d(x) = d$, complete steps \ref{step:mlow} through \ref{step:mhigh-add}.
    \label{step:induct}

    \Statex
        \begin{quote}$\rhd$
        \textit{
        Inductive Case: compute $L(x)$ layer by layer towards $s$.
        }
        \end{quote}
        \hfill

    \Indent
        \State 
        \parbox[t]{367pt}{
        For each integer $m$ with $0\le m < \alpha k$, and for each choice of integers $k_1, \ell_2\ge 0$ satisfying $k_1 + \ell_2\le (3k+m+2)/4$, complete steps \ref{step:mlow-1} and \ref{step:mlow-2}.} \label{step:mlow}

            \Statex
            \begin{quote}
            \begin{quote}\; 
            $\rhd$
        \textit{
        This step handles \textbf{Case 1:} the solution path has few stable edges. 
        }
        \end{quote}
        \end{quote}
        \hfill

            \Indent

                \State  \parbox[t]{352pt}{If there is a path of length $\ell\le 2 k_1 + \ell_2$ from $x$ to $t$ in $\G{x}$, update $L(x)\leftarrow L(x)\cup\set{\ell}$.\strut}
                \label{step:mlow-1}
                
                \Statex
                \begin{quote}
                \begin{quote}
                \begin{quote}\;\;
                $\rhd$
                \textit{
                This step handles \textbf{Case 1(a):} $d(t) - d(x)\le (k-m)/2$.
                }
                \end{quote}
                \end{quote}
                \end{quote}
                \hfill
    
                \State \parbox[t]{352pt}{Try out all vertices $y$ with $d(y)\in [d(x)+1,\min(d(t), d(x) + (3k-m)/2+1)]$. If for some such $y$, there is a path from $x$ to $y$ in $\Gint{x}{y}$ of length $a\le 2 k_1 + \ell_2$ with exactly $k_1$ vertices in $V_1$, and $\ell_2$ edges with both endpoints in $V_2$, update $L(x)\leftarrow L(x)\cup\grp{a + L(y)}$.\strut}
                \label{step:mlow-2}

                \Statex
                \begin{quote}
                \begin{quote}
                \begin{quote}\;\;
                $\rhd$
                \textit{
                This step handles \textbf{Case 1(b):} $d(t) - d(x) > (k-m)/2$.
                }
                \end{quote}
                \end{quote}
                \end{quote}
                \hfill
            \EndIndent

        \State For each integer $m$ with $\alpha k < m \le k$, complete step \ref{step:mhigh-add}.\label{step:mhigh}

                \Statex
        \begin{quote}
        \begin{quote}\;
        $\rhd$
        \textit{
        This step handles \textbf{Case 2:} the solution path has many stable edges. 
        }
        \end{quote}
        \end{quote}
        \hfill

        \Indent

            \State \parbox[t]{352pt}{Try out all vertices $y$ with $d(y)\in [d(x)+1, d(x)+(1-\alpha)k/2+1]$.
            If for some such $y$, there is a path from $x$ to $y$ in $\Gint{x}{y}$ of length $a\le (3-\alpha)k/2 + 1$, update $L(x)\leftarrow L(x)\cup\grp{a + L(y)}$\strut} \label{step:mhigh-add}
    
        \EndIndent
    \EndIndent

    \State Return \textbf{yes} if and only if $(\dist(s,t) + k)\in L(s)$. \label{step:answer}

    \end{algorithmic}

\end{algorithm}

Our algorithm computes a set $L(x)$ for each vertex $x$ in the graph, corresponding to the possible lengths of potential subpaths from $x$ to $t$ of a solution path from $s$ to $t$.

In step \ref{step:base} of \Cref{alg:detour}, we compute $L(x)$ for all $x$ that are ``far'' from $s$, by solving instances of \textsf{$\ell$-Path} for $\ell\le (3-\alpha)k/2$.
Starting in step \ref{step:induct}, we compute $L(x)$ for vertices $x$ closer to $s$, in terms of the previously computed sets $L(y)$ for vertices $y$ further from $s$. 
In steps \ref{step:mlow} through \ref{step:mlow-2}, we compute some lengths in $L(x)$ by solving instances of \textsf{$(a, k_1, \ell_2)$-Bipartitioned Path} for appropriate $a, k_1, \ell_2$ values, and in \ref{step:mhigh} and \ref{step:mhigh-add} we compute the remaining lengths in $L(x)$ by solving \textsf{$a$-Path} for $a\le (3-\alpha)k/2 + 1$.

\subsection{Correctness}\label{subsec:correctness}

In this section, we show that \Cref{alg:detour} correctly solves the \textsf{$k$-Detour} problem for any choice of $\alpha\in (0,1)$.
The main technical part of the proof lies in inductively showing that every possible solution path from $s$ to $t$ will be considered by the algorithm and its length will be included in the set $L(s)$. 
In \Cref{alg:detour}, we try out values of the variable $m$ from $0$ to $k$, and execute differently depending on how $m$ compares to $\alpha k$. 
This is interpreted as follows: suppose there is a solution path $P$ from $x$ to $t$, then $m$ corresponds to a  guess of the number of stable edges in $P$. 

In \textbf{Case 1}, we guess that $P$ has few stable edges $m < \alpha k$ which corresponds to steps~\ref{step:mlow} to \ref{step:mlow-2}. 
Under \textbf{Case 1}, there are two possible structures a potential solution path might take on depending on how $d(x)$ compares to $d(t)$. 
We refer to the case where $d(x) - d(t)$ is small as \textbf{Case 1(a)} considered by step~\ref{step:mlow-1}, and the case where $d(x) - d(t)$ is large as \textbf{Case 1(b)} considered by step~\ref{step:mlow-2}. In \textbf{Case 2}, we guess that $m \ge \alpha k$, so $P$ has many stable edges,  which corresponds to steps~\ref{step:mhigh} to \ref{step:mhigh-add}. These cases are also formally defined in our proof of correctness.

\begin{theorem}
    \label{thm:detour-correct}
    For any fixed $\alpha \in (0,1)$, \Cref{alg:detour} correctly solves the \textsf{$k$-Detour} problem.
\end{theorem}
\begin{proof}

We prove that upon halting, each set $L(x)$ computed by \Cref{alg:detour} has the property that for all integers $\ell\in [d(t)-d(x), d(t)-d(x)+k]$, we have 
    \begin{equation}
    \label{key-property}
\ell\in L(x) \text{ if and only if there is a path of length }\ell\text{ from }x\text{ to }t\text{ in }\G{x}.
    \end{equation}
If this property holds, then step \ref{step:answer} of \Cref{alg:detour} returns the correct answer to the \textsf{$k$-Detour} problem, since $\dist(s,t) + k$ is in $L(s)$ if and only if there is a path from $s$ to $t$ of length $\dist(s,t) + k$ in $\G{s} = G$.

So, it suffices to show that \cref{key-property} holds for all vertices $x$.
We prove this result by induction on the distance of $x$ from $s$ in the graph.

\noindent \textbf{Base case:} For the base case, suppose $x$ is a vertex with
    \begin{equation}
    \label{eq:x-basecase}
    d(x)\in [d(t)-(1-\alpha)k/2, d(t)].
    \end{equation}
Then $L(x)$ is computed in step \ref{step:base} of \Cref{alg:detour}.
We now verify that \cref{key-property} holds.

First, suppose $\ell\in L(x)$.

Then, $\ell$ must be the length of some path from $x$ to $t$ in $\G{x}$ by design.

Conversely, suppose we have a path $P$ from $x$ to $t$ in $\G{x}$ of some length \[\ell\le d(t)-d(x) + k.\]
Then by the assumption on $x$ from \cref{eq:x-basecase} in this case, we have
    \[\ell \le d(t) - d(x) + k \le (1-\alpha)k/2 + k = (3-\alpha)k/2\]
so step \ref{step:base} of \Cref{alg:detour} correctly includes $\ell$ in $L(x)$.

Thus \Cref{key-property} holds for all vertices $x$ satisfying \cref{eq:x-basecase}.

\noindent \textbf{Inductive case:} For the inductive step, suppose $x$ is a vertex with
\begin{equation}
\label{eq:x-induct}
d(x) \le d(t) - (1-\alpha)k/2 - 1.
\end{equation}
We may inductively assume that we have computed sets $L(y)$ satisfying \cref{key-property}, for all vertices $y$ with $d(y) > d(x)$.

Suppose $\ell\in L(x)$ at the end of \Cref{alg:detour}.
Then either $\ell$ was added to $L(x)$ in step \ref{step:mlow-1}, or $\ell$ was added to $L(x)$ in steps \ref{step:mlow-2} or \ref{step:mhigh-add} of \Cref{alg:detour}.
In the former case, $\ell$ is the length of a path from $x$ to $t$ in $\G{x}$ by design.
In the latter cases, we have $\ell = a + b$, where $a$ is the length of some path from $x$ to $y$ (for some vertex $y$ with $d(y) > d(x)$) in $\Gint{x}{y}$, and (by the inductive hypothesis) $b$ is the length of some path from $y$ to $t$ in $\G{y}$.
Since $\Gint{x}{y}$ and $\G{y}$ intersect only at $y$, the union of these paths is a path from $x$ to $t$ in $\G{t}$.
So, every integer in $L(x)$ is a valid length of a path from $x$ to $t$ in $\G{x}$ as desired.

Conversely, suppose there is a path $P$ from $x$ to $t$ in $\G{x}$ of length 
    \begin{equation}
    \label{eq:path-length}
    \ell \in [d(t)-d(x), d(t)-d(x)+k].
    \end{equation}
We prove that $\ell$ appears in $L(x)$.

To do this, we will analyze the number of forward, backward, and stable edges appearing in $P$.
Note that $P$ has at least $d(t)-d(x)$ forward edges, since $P$ begins at a vertex at distance $d(x)$ from $s$, ends at a vertex at distance $d(t)$ from $s$, and only the forward edges allow us to move to vertices further from $s$.

Let $m$ denote the number of stable edges in $P$.
We have $m\le k$, since the length of $P$ is at most $d(t)-d(x) + k$, and $P$ has at least $d(t)-d(x)$ forward edges.

\begin{claim}
    \label{claim:split}
    Suppose $d(x) \le d(t) - (k-m)/2 - 1$.
    Then $P$ contains a vertex $y$ such that

    \begin{enumerate}
        \item
            $d(y)\in [d(x)+1,d(x) + (k-m)/2+1]$,
        \item 
            every vertex $u\in P[y,t]$ with $u\neq y$ has $d(u) > d(y)$, and
        \item
            every vertex $v\in P[x,y]$ has $d(v)\le d(y)$.
    \end{enumerate}
\end{claim}

\begin{claimproof}
    For each $i\in [(k-m)/2+1]$, let $z_i$ denote the last vertex on $P$ satisfying
        \[d(z_i) = d(x) + i.\]
    These vertices exist because we are assuming that $d(x)\le d(t) - (k-m)/2-1$, and $P$ must contain vertices $v$ with $d(v) = d$ for every $d\in [d(x), d(t)]$.

    By definition, each $z_i$ satisfies conditions \textbf{1} and \textbf{2} from the claim.
    If some $z_i$ satisfies condition \textbf{3} as well, then the claim is true.

    So, suppose that none of the $z_i$ satisfy condition \textbf{3}.
    This means that for each index $i$, the subpath $P[x,z_i]$ contains a vertex $u$ with $d(u) > d(z_i)$.
    Consecutive vertices in $P$ differ in their distance from $s$ by at most one, so $P[x,z_i]$ must contain an edge $e = (v,w)$ such that $d(v) = d(w)+1$ and $d(w) = d(z_i) = d(x) + i$.
    That is, $P$ contains a backwards edge from a vertex at distance $i+1$ from $s$ to a vertex at distance $i$ from $s$, as depicted in \Cref{fig:back-edges}.

    Note that $z_1, z_2, \dots, z_{(k-m)/2+1}$ occur on $P$ in the listed order.
    This is because
        \[d(z_1) < d(z_2) < \dots < d(z_{(k-m)/2+1})\]
    and each $z_i$ satisfies condition \textbf{1} from the claim.
    Combined with the discussion from the previous paragraph, this means that $P$ contains at least $(k-m)/2+1$ backwards edges. 
    We now argue that this violates the assumption on the length of $P$.

    Let $f$ and $b$ denote the number of forward and backwards edges in $P$ respectively.
    
    Since $P$ starts at $x$ and ends at $t$, we have 
    $f - b = d(t) - d(x)$, which implies that 
        \begin{equation}
            \label{eq:netchange}
            f = d(t) - d(x) + b.
        \end{equation}
    Then the total length of $P$ is 
        $f + b + m = d(t) - d(x) + m + 2b$
    by \cref{eq:netchange}.
    However, since $P$ has at least $(k-m)/2+1$ backwards edges, this length satisfies 
    \[d(t) - d(x) + m + 2b\ge d(t) - d(x) + m + 2\grp{(k-m)/2 + 1} > d(t) - d(x) + k\]
    which contradicts the fact that the length $\ell$ of $P$ satisfies \cref{eq:path-length}.
    Thus our assumption was incorrect, and one of the $z_i$ satisfies all three conditions from the claim, as desired.
\end{claimproof}

    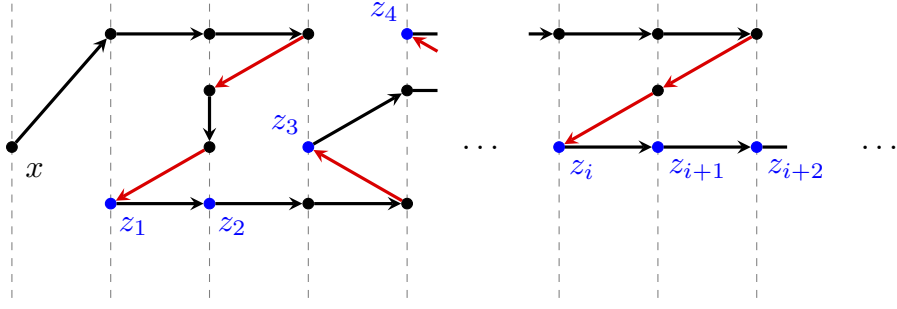
\begin{figure}
        \centering
            \begin{tikzpicture}[
    vtx/.style={circle,draw,fill=black,inner sep=0pt,minimum width=4pt},
    edge/.style={-stealth, very thick}, b/.style={-stealth, very thick, color=red!85!black}
    ]

    \def\top{4.5cm};
    \def\bottom{0.5cm};
    \def\half{{(0.5*(\top+\bottom))}};
    
    \def\start{0.25cm};
    \def\layershift{1.3cm};
    \def\smallshift{0.7cm};

    \foreach \x in {0,...,4} {
        \draw[gray,dashed] (\start + \layershift*\x,\bottom) -- (\start + \layershift*\x,\top);
    }

    \foreach \x in {5,...,7} {
        \draw[gray,dashed] (\start + \smallshift + \layershift*\x,\bottom) -- (\start + \smallshift + \layershift*\x,\top);
    }

    \node[scale=1.2] at (\start + 4.5*\layershift + 0.5*\smallshift, 2.5) {$\dots$};
    \node[scale=1.2] at (\start + 8*\layershift + 1.5*\smallshift, 2.5) {$\dots$};
    \node[scale=1.2] at (\start - \smallshift, 2.5) {$\phantom{\dots}$};

    \node[vtx] (1) at (\start,\half) {};
    \node[scale=1.2] at (\start + 0.3cm, 2.2cm) {$x$};

    \node[vtx] (2) at (\start + \layershift, 4) {};
    \node[vtx,color=blue] (7) at (\start + \layershift, 1.75) {};
    \node[scale=1.2,color=blue] at (\start + \layershift + 0.3cm, 1.45cm) {$z_1$};

    \node[vtx] (3) at (\start + 2*\layershift, 4) {};
    \node[vtx] (5) at (\start + 2*\layershift, 3.25) {};
    \node[vtx] (6) at (\start + 2*\layershift, 2.5) {};
    \node[vtx, color=blue] (8) at (\start + 2*\layershift, 1.75) {};
    \node[scale=1.2,color=blue] at (\start + 2*\layershift + 0.3cm, 1.45cm) {$z_2$};

    \node[vtx] (4) at (\start + 3*\layershift, 4) {}; \node[vtx] (9) at (\start + 3*\layershift, 1.75) {};
    \node[vtx,color=blue] (11) at (\start + 3*\layershift, 2.5) {};
    \node[scale=1.2,color=blue] at (\start + 3*\layershift - 0.3cm, 2.8cm) {$z_3$};

    \node[vtx] (10) at (\start + 4*\layershift, 1.75) {};
    \node[vtx] (12) at (\start + 4*\layershift, 3.25) {};
    \node[vtx,color=blue] (13) at (\start + 4*\layershift, 4) {};
    \node[scale=1.2,color=blue] at (\start + 4*\layershift - 0.3cm, 4.3cm) {$z_4$};

    \node[vtx] (14) at (\start + \smallshift + 5*\layershift, 4) {};
    \node[vtx,color=blue] (18) at (\start + \smallshift + 5*\layershift, 2.5) {};
    \node[scale=1.2,color=blue] at (\start + \smallshift + 5*\layershift + 0.3cm, 2.2) {$z_{i}$};

    \node[vtx] (15) at (\start + \smallshift+ 6*\layershift, 4) {};
    \node[vtx] (17) at (\start + \smallshift + 6*\layershift, 3.25) {};
    \node[vtx,color=blue] (19) at (\start + \smallshift + 6*\layershift, 2.5) {};
    \node[scale=1.2,color=blue] at (\start + \smallshift + 6*\layershift + 0.52cm, 2.2) {$z_{i+1}$};

        \node[vtx] (16) at (\start + \smallshift + 7*\layershift, 4) {};
        \node[vtx,color=blue] (20) at (\start + \smallshift + 7*\layershift, 2.5) {};
    \node[scale=1.2,color=blue] at (\start + \smallshift + 7*\layershift + 0.52cm, 2.2) {$z_{i+2}$};

    \draw[edge]   (1) -- (2);
    \draw[edge]   (2) -- (3);
    \draw[edge]   (3) -- (4);
    \draw[b]      (4) -- (5);
    \draw[edge]   (5) -- (6);
    \draw[b]      (6) -- (7);
    \draw[edge]   (7) -- (8);
    \draw[edge]   (8) -- (9);
    \draw[edge]  (9) -- (10);
    \draw[b]    (10) -- (11);
    \draw[edge] (11) -- (12);
    \draw[edge] (14) -- (15);
    \draw[edge] (15) -- (16);
    \draw[b]    (16) -- (17);
    \draw[b]    (17) -- (18);
    \draw[edge] (18) -- (19);
    \draw[edge] (19) -- (20);

    \draw[very thick] (12) -- (\start + 4*\layershift + 0.4cm, 3.25);
    \draw[b] (\start + 4*\layershift + 0.4cm, 4cm - 0.2307cm) -- (13);
    \draw[very thick] (13) -- (\start + 4*\layershift + 0.4cm, 4cm);
    \draw[-stealth, very thick] (\start + \smallshift + 5*\layershift - 0.4cm, 4cm) -- (14);
    \draw[very thick] (20) -- (\start + \smallshift + 7*\layershift + 0.4cm, 2.5cm);
    
    \end{tikzpicture}
        \caption{If node $z_i$ does not satisfy condition \textbf{3} of \Cref{claim:split}, it means that before hitting $z_{i}$, the path visited a node further from $s$ than $z_{i}$. Thus, we can associate $z_i$ with some backwards edge.
        The presence of too many of these backwards edges would violate the length condition on $P$, so it turns out that one such node (in the figure, $z_{i+2}$) does have to satisfy condition \textbf{3}.}
        \label{fig:back-edges}
    \end{figure}

We now perform casework on the number of stable edges $m$ in $P$. We start with \textbf{Case 2} from step \ref{step:mhigh} of \Cref{alg:detour}, since this is the easiest case to analyze.

\paragraph*{Case 2: Many Stable Edges ($\boldsymbol{m \ge \alpha k}$)}
Suppose $m \ge \alpha k$.
In this case, by \cref{eq:x-induct} we have 
    \[d(x) \le d(t) - (1-\alpha)k/2 - 1 \le d(t) - (k-m)/2 - 1\]
So, by \Cref{claim:split}, there exists a vertex $y$ in $P$ satisfying the three conditions of \Cref{claim:split}.

By condition \textbf{3} from \Cref{claim:split}, the subpath
    $A = P[x,y]$
is contained in $\Gint{x}{y}$.
By condition \textbf{2} from \Cref{claim:split}, the subpath
    $B = P[y,t]$
is contained in $\G{y}$.

Let $a$ denote the length of $A$, and $b$ denote the length of $B$.

Since $A$ has length at least $d(y) - d(x)$, and $P$ has length at most $d(t)-d(x) + k$ by \cref{eq:path-length}, we know that the length $B$ satisfies
    \begin{equation}
    \label{case1-bbound}
    b \le d(t) - d(y) + k.
    \end{equation}
By the inductive hypothesis, $L(y)$ satisfies \cref{key-property}, so $b\in L(y)$.

Similar to the reasoning that established \cref{case1-bbound}, we can prove that
    \begin{equation}
    \label{case1-abound}
    a \le d(y) - d(x) + k.
    \end{equation}
By condition \textbf{1} of \Cref{claim:split}, we
know that $d(y)\le d(x) + (k-m)/2 + 1$.
Since $m \ge \alpha k$, this implies that 
    $d(y)\le d(x) + (1-\alpha)k/2 + 1.$
    
Substituting this into \cref{case1-abound} yields
    \[a\le (1-\alpha)k/2 + 1 + k = (3-\alpha)k/2 + 1.\]
Thus, the length $a$ of $A$ will be found in step \ref{step:mhigh-add} of \Cref{alg:detour}.
As mentioned before, $b\in L(y)$.
Thus, $\ell = a+b\in (a + L(y))$ is correctly added to the set $L(x)$ in step \ref{step:mhigh-add} of \Cref{alg:detour}, which proves the desired result in this case.

\paragraph*{Case 1: Few Stable Edges ($\boldsymbol{m < \alpha k}$)}

If we do not fall into {\bf Case 2}, we must have $m < \alpha k$.
Recall that in step \ref{step:partition} of \Cref{alg:detour}, we defined 
    $V_1 = \set{u\mid d(u)\text{ is odd}}$
and 
    $V_2 = \set{u\mid d(u)\text{ is even}}.$

We want to argue that most edges in path $P$ cross the bipartition $V_1\sqcup V_2$.
To that end, the following claim will be helpful.

\begin{restatable}{claim}{labelbounding}
    \label{claim:label-bounding}
    Let $Q$ be a path of length $q$, with at most $m$ stable edges.
    Let $k_1$ denote the number of vertices of $Q$ in $V_1$, and let $\ell_2$ denote the number of edges in $Q$ with both endpoints in $V_2$.
    Then we have
        \[k_1 + \ell_2 \le (q + m + 1)/2.\]
\end{restatable}
\begin{claimproof}
Let $k_2$ denote the number of vertices of $Q$ in $V_2$.

Consider the cycle $C$ formed by taking $Q$ together with an additional edge between its endpoints (this new edge is imagined for the purpose of argument, and does not change the definition of $V_1$ and $V_2$).

Let $q_1$, $q_2$, and $q_{\t{cross}}$ denote the number of edges of $C$ with both endpoints in $V_1$, both endpoints in $V_2$, and endpoints in both $V_1$ and $V_2$ respectively. 
We have 
    \begin{equation}
    \label{eq:double-count1}
    2k_1 = 2q_1 + q_{\t{cross}}
    \end{equation}
because both sides of the above equation count the number of pairs $(u, e)$ such that $u$ is a vertex in $C\cap V_1$, and $e$ is an edge in $C$ incident to $u$. 
A symmetric argument implies that
    \begin{equation}
    \label{eq:double-count2}
    2q_2 + q_{\t{cross}} = 2k_2.
    \end{equation}
Adding \Cref{eq:double-count1} and \Cref{eq:double-count2} together and simplifying yields
    \[k_1 + q_2 = k_2 + q_1.\]
This implies that 
\[
    k_1 + q_2 = \grp{k_1 + k_2 + q_1 + q_2}/2.
\]
Since $C$ is $Q$ with one additional edge, we have $\ell_2 \le q_2$.
So the above equation implies that 
    \begin{equation}
    \label{eq:prelim-bound}
    k_1 + \ell_2 \le \grp{k_1 + k_2 + q_1 + q_2}/2.
    \end{equation}
We have 
\begin{equation}
\label{eq:vertices-length}
k_1 + k_2 = q+1
\end{equation}
since the total number of vertices in $Q$ must be one more than its length. 
By assumption on the number of stable edges in $Q$, we have 
\begin{equation}
\label{eq:stable-bound}
q_1 + q_2 \le m.
\end{equation}
Substituting \cref{eq:vertices-length} and \cref{eq:stable-bound} into the right hand side of \cref{eq:prelim-bound} yields
    \[k_1 + \ell_2 \le (q + m  + 1)/2\]
which proves the desired result. 
\end{claimproof}
\noindent With \Cref{claim:label-bounding} established, we are now ready to analyze the two subcases under \textbf{Case 1}, based on the relative distances of $x$ and $t$ from $s$. 

\paragraph*{Case 1(a): $\boldsymbol{d(t) - d(x)}$ is small}
Suppose that $d(x) \in [d(t)-(k-m)/2, d(t)].$

In this case, \cref{eq:path-length} implies that $P$ has length 
    \[\ell \le d(t) - d(x) + k \le (3k-m)/2.\]
Let $k_1$ denote the number of vertices of $P$ in $V_1$, and $k_2$ denote the number of edges in $P$ with both endpoints in $V_2$.
Then by setting $Q = P$ and $q = \ell$ in \Cref{claim:label-bounding}, we have 
    \begin{equation}
    \label{2a-label-bound}
    k_1 + \ell_2 \le (\ell + m + 1)/2 \le (3k+m+2)/4.
    \end{equation}
Also, note that $P$ has length $\ell \le 2 k_1 + \ell_2$, since $2k_1$ is greater than or equal to the number of edges in $P$ incident to a vertex in $V_1$.
This observation, together with \cref{2a-label-bound}, shows that in this case, the length $\ell$ is correctly included in $L(x)$ in step \ref{step:mlow-1} of \Cref{alg:detour}.

\paragraph*{Case 1(b): $\boldsymbol{d(t) - d(x)}$ is large}

If we do not fall into \textbf{Case 1(a)}, it means that 
\begin{equation}
\label{2a-xinduct}
d(x)\le d(t) - (k-m)/2 - 1.
\end{equation}
Thus, by \Cref{claim:split}, there exists a vertex $y$ in $P$ satisfying the three conditions of \Cref{claim:split}.
The proof that $\ell\in L(x)$ in this case is essentially a combination of the  proofs from \textbf{Case 2} and \textbf{Case 1(a)}.

As in \textbf{Case 2}, by condition \textbf{3} from \Cref{claim:split}, the subpath
    $A = P[x,y]$
is contained in $\Gint{x}{y}$.
By condition \textbf{2} from \Cref{claim:split}, the subpath
    $B = P[y,t]$
is contained in $\G{y}$.

Let $a$ and $b$ denote the lengths of paths $A$ and $B$ respectively.
Reasoning identical to the arguments which established \cref{case1-bbound,case1-abound} prove that in this case we also have 
    \begin{equation}
    \label{case2b-bbound}
    b\le d(t) - d(y) + k
    \end{equation}
    and
    \begin{equation}
    \label{case2b-abound}
    a\le d(y) - d(x) + k.
    \end{equation}
Condition \textbf{1} of \Cref{claim:split} implies that $d(y)\le d(x) + (k-m)/2 + 1$.
Substituting this into \cref{case2b-abound} implies that
    \[a\le (3k-m)/2 + 1.\]
Let $k_1$ denote the number of vertices of $A$ in $V_1$, and let $\ell_2$ denote the number of edges in $A$ with both endpoints in $V_2$.
Then by setting $Q = A$ and $q = a$ in \Cref{claim:label-bounding}, we have 
\begin{equation}
\label{2b-label-bound}
    k_1 + \ell_2 \le (a + m + 1)/2 \le (3k + m + 2)/4.
\end{equation}
Also, we know that $a\le 2k_1 + \ell_2$, because $2k_1$ is greater than or equal to the number of edges in $A$ incident to a vertex in $V_1$.
Combining this observation with \cref{2b-label-bound}, we see that the length $a$ is indeed computed in step \ref{step:mlow-2} of \Cref{alg:detour}.

By the inductive hypothesis (\cref{key-property}) and \cref{case2b-bbound}, we know that $b\in L(y)$.
Thus we have $\ell = a+b \in (a+L(y))$, so in this case, $\ell$ is correctly included in $L(x)$ in step \ref{step:mlow-2} of \Cref{alg:detour}.

This completes the induction, and proves that \cref{key-property} holds for all vertices $x$ in the graph.
In particular, \cref{key-property} holds for $x$ equal to $s$.
This implies that step \ref{step:answer} of \Cref{alg:detour} returns the correct answer to the \textsf{$k$-Detour} algorithm.
\end{proof}

\section{Applications}\label{subsec:app}

In this section, we present consequences of our new algorithm for \textsf{$k$-Detour} from \Cref{sec:detour}.



\detourthm*
\begin{proof}
By \Cref{thm:detour-correct}, \Cref{alg:detour} correctly solves \textsf{$k$-Detour}, for any value $\alpha\in (0,1)$, 

What is the runtime of \Cref{alg:detour}?
Well, steps \ref{step:base} and \ref{step:mhigh-add} of \Cref{alg:detour} involve solving polynomially many instances of \textsf{$\ell$-Path}, for $\ell\le (3-\alpha)k/2 + 1$.
Using the fastest known algorithm for \textsf{$k$-Path} in undirected graphs \cite{narrow-sieves}, these steps take
    \[1.657^{(3-\alpha)k/2}\poly(n)\]
time.
The remaining computationally intensive steps of \Cref{alg:detour} occur in steps \ref{step:mlow-1} and \ref{step:mlow-2}, which can be implemented by solving $\poly(n)$ instances of \textsf{$(\ell, k_1, \ell_2)$-Bipartitioned Path}, for $k_1 + \ell_2 < (3k + \alpha k + 2)/4$.
By \Cref{lem:near-bipartite}, these steps then take
    \[2^{(3+\alpha)k/4}\poly(n)\]
time overall.
Then by setting $\alpha = 0.55814$ to balance the above runtimes, we see that we can solve \textsf{$k$-Detour} over undirected graphs in 
    \[\grp{1.657^{(3-\alpha)k/2} + 2^{(3+\alpha)k/4}}\poly(n) \le 1.8526^k\poly(n)\]
time, as desired. \qedhere
\end{proof}

\longestdetour*
\begin{proof}
The proof of \cite[Corollary 2]{directed-detours} shows that \textsf{$k$-Longest Detour} in undirected graphs reduces, in polynomial time, to solving \textsf{$p$-Detour} for all $p\le 2k$ and $\poly(n)$ instances of \textsf{$(3k/2)$-Longest Path} on graphs with at most $n$ nodes.

The proof of \Cref{thm:k-detour} implies that \textsf{$k$-Detour} can be solved over undirected graphs in $1.8526^k\poly(n)$ time.
Previous work in \cite[Section 6.3]{det-sieving} shows that \textsf{$k$-Longest Path} can be solved over undirected graphs in $1.657^k\poly(n)$ time.
 Combining these results together with the above discussion shows that \textsf{$k$-Longest Detour} can be solved over undirected graphs in
    \[\grp{1.8526^{2k} + 1.657^{3k/2}}\poly(n) \le 3.432^k\poly(n) \]
time, as desired.
\end{proof}

\detdetour*
\begin{proof}
By \Cref{thm:detour-correct}, we can solve \textsf{$k$-Detour} over an undirected graph by running \Cref{alg:detour} with parameter $\alpha = 0$.
When $\alpha = 0$ in \Cref{alg:detour}, steps \ref{step:mlow}, \ref{step:mlow-1}, \ref{step:mlow-2} never occur. 
In this case, the algorithm only needs to solve $\poly(n)$ instances of \textsf{$\ell$-Path}, for $\ell \le 3k/2 + 1$, in steps \ref{step:base} and \ref{step:mhigh-add}.
Since \textsf{$k$-Path} can be solved deterministically in $2.554^k\poly(n)$ time \cite{fast-det-kpath}, this means that  we can solve \textsf{$k$-Detour} deterministically in 
    \[2.554^{3k/2}\poly(n) \le 4.0817^k\poly(n)\]
time, as desired.
\end{proof}

\detlongestdetour*
\begin{proof}
The proof of \cite[Corollary 2]{directed-detours} shows that \textsf{$k$-Longest Detour} in undirected graphs reduces, in deterministic polynomial time, to solving \textsf{$p$-Detour} for $p\le 2k$, and $\poly(n)$ instances of \textsf{$(3k/2)$-Longest Path} on graphs with at most $n$ nodes.

The proof of \Cref{det-detour} implies that \textsf{$k$-Detour} can be solved over undirected graphs deterministically in $4.0817^k\poly(n)$ time.
Previous work \cite{longest-path} shows that \textsf{$k$-Longest Path} can be solved deterministically in $4.884^k\poly(n)$ time.
Combining these results together with the above discussion shows that \textsf{$k$-Longest Detour} can be solved over undirected graphs deterministically in 
    \[\grp{4.0817^{2k} + 4.884^{3k/2}}\poly(n)\le 16.661^k\poly(n)\]
time, as desired.
\end{proof}

\section{Conclusion}
\label{sec:conclusion}

In this paper, we obtained faster algorithms for \textsf{$k$-Detour} and \textsf{$k$-Longest Detour} over undirected graphs.
However, many mysteries remain surrounding the true time complexity of these problems.
We highlight some open problems of interest, relevant to our work.

\begin{enumerate}
\item
    The most pertinent question: what is the true parameterized time complexity of \textsf{$k$-Detour} and \textsf{$k$-Longest Detour}? 
    In particular, could it be the case that \textsf{$k$-Detour} can be solved as quickly as \textsf{$k$-Path}, and \textsf{$k$-Longest Detour} can be solved as quickly as \textsf{$k$-Longest Path}?
    No known conditional lower bounds rule out these possibilities.

\item
    The current fastest algorithm for \textsf{$k$-Longest Path} in directed  graphs has a bottleneck of solving \textsf{$2k$-Path}.
    The current fastest algorithm for \textsf{$k$-Detour} in directed graphs has a bottleneck of solving \textsf{$2k$-Path}.
    Similarly, the fastest known algorithm\footnote{In fact, even the recent alternate algorithm of \cite{detour-2dp} for \textsf{$k$-Longest Detour} requires solving \textsf{$2k$-Detour} first.} for \textsf{$k$-Longest Detour} in undirected graphs requires first solving \textsf{$2k$-Detour}.
    Is this parameter blow-up necessary? 
    Could it be possible to solve these harder problems with parameter $k$ faster than solving these easier problems with parameter $2k$? 

\item
    The speed-up in our results crucially uses a fast algorithm for the \textsf{$(\ell,k_1, \ell_2)$-Bipartitioned Path} problem in undirected graphs.
    In directed graphs no $(2-\varepsilon)^{\ell}\poly(n)$ time algorithm appears to be known for this problem, for any constant $\varepsilon > 0$ and interesting range of parameters $k_1$ and $\ell_2$.
    Such improvements could yield faster algorithms for \textsf{$k$-Detour} in directed graphs.
    Can we get such an improvement? 
    Also of interest: can we get a faster deterministic algorithm for \textsf{$(\ell,k_1, \ell_2)$-Bipartitioned Path}?
\item
    An easier version of the previous question, also raised in \cite[Section 9.1]{det-sieving}: can we solve \textsf{$k$-Path} in directed bipartite graphs in $(2-\varepsilon)^k\poly(n)$ time, for some constant $\varepsilon > 0$? 
    In the unparameterized setting, the \textsf{Hamiltonian Path} (\textsf{$k$-Path} for $k=n$) problem admits several distinct algorithms running in $(2-\varepsilon)^n\poly(n)$ time in directed bipartite graphs.
    Specifically, \cite{base-matching-hamiltonian} shows \textsf{Hamiltonian Path} in directed bipartite graphs can be solved in $1.888^n\poly(n)$ time, and \cite{hamiltonian-matrix-tree-theorem} uses very different methods to solve this problem even faster in $3^{n/2}\poly(n)$ time.\footnote{It is also known that sufficient improvements to algorithms for multiplying two $n\times n$ matrices together would imply that even the weighted version of \textsf{Hamiltonian Path} in directed bipartite graphs can be solved in $(2-\varepsilon)^n\poly(n)$ time, for some constant $\varepsilon > 0$ \cite{nederlof-life}.}
    We conjecture that the same speed-up is possible for \textsf{$k$-Path}, so that this problem can be solved over directed bipartite graphs in $3^{k/2}\poly(n)$ time.

\end{enumerate}

\bibliography{main}

\appendix

\section{Additional Related Work}
\label{sec:related-work}


\paragraph*{Detours in Directed Graphs}
The \textsf{$k$-Longest Detour} problem is not known to be \textsf{FPT} over directed graphs.
However, this problem has been proven to be \textsf{FPT} for certain restricted classes of graphs.
For example, the algorithm of \cite{directed-detours} shows that \textsf{$k$-Longest Detour} is \textsf{FPT} on any class $\mathcal{G}$ of graphs where the \textsf{3-Disjoint Paths}\footnote{For any positive integer $p$, in the \textsf{$p$-Disjoint Paths} problem, we are given a graph along with source vertices $s_1, \dots, s_p$ and target vertices $t_1, \dots, t_p$, and tasked with finding disjoint paths from $s_i$ to $t_i$ for each index $i$. For constant $p$, this problem can be solved in polynomial time over undirected graphs. Already for $p=2$, this problem is \textsf{NP}-hard over general directed graphs.} problem can be solved in polynomial time.
This implies, for example, that \textsf{$k$-Longest Detour} is \textsf{FPT} over directed planar graphs
(see also \cite{planar-detour}, which presents a more direct argument showing \textsf{$k$-Longest Detour} is \textsf{FPT} in directed planar graphs).
More recently, \cite{detour-2dp} showed that \textsf{$k$-Longest Detour} is \textsf{FPT} on any  class $\mathcal{G}$ of graphs where the \textsf{2-Disjoint Paths} problem can be solved in polynomial time. 

\paragraph*{Deterministic Algorithms}
There has been significant research into obtaining fast deterministic algorithms for path and detour problems.
The fastest known deterministic algorithm for \textsf{$k$-Path} runs in $2.554^k\poly(n)$ time \cite{fast-det-kpath}, and the fastest known deterministic algorithm for \textsf{$k$-Longest Path} runs in $4.884^k\poly(n)$ time \cite{longest-path} (the runtimes reported in \Cref{sec:intro} for these problems come from randomized algorithms).
The work of \cite{detour-original} implies that the fastest known deterministic algorithm for \textsf{$k$-Detour} runs in $6.523^k\poly(n)$ time \cite{longest-path}.
Interestingly, for the \textsf{$k$-Path}, \textsf{$k$-Longest Path}, and \textsf{$k$-Detour} problems, no faster deterministic algorithms are known  for the special case of undirected graphs.
The work of \cite{directed-detours} implies that the fastest known deterministic algorithm for \textsf{$k$-Longest Detour} in undirected graphs runs in $42.549^k\poly(n)$ time.

\paragraph*{Above-Guarantee Parameterization}

The study of the \textsf{$k$-Detour} and \textsf{$k$-Longest Detour} problems belongs to a large subarea of parameterized algorithms which focuses on ``above-guarantee parameterizations.''
In these problems, we are given an input which is guaranteed to contain a structure of some size $\sigma$, and our task is to determine if the input contains a similar structure of size $k$ ``more than'' $\sigma$ (the definition of this increase in size depending on the problem of interest).
In the detour problems we discuss, we are guaranteed a path of length $\dist(s,t)$ from $s$ to $t$.
We refer the reader to \cite{above-guarantee-survey} for an accessible survey of results and open problems in this area.

\paragraph*{Conditional Lower Bounds}

In this section, we first define some problems, recall popular conjectures concerning the exact time complexity of those problems, and then state implications of those conjectures for the exact time complexity of variants of the \textsf{$k$-Path} problem.

In the \textsf{3SAT} problem, we are given a \textsf{3-CNF} $\varphi$ (i.e., a Boolean formula which can be written as a conjunction of clauses, where each clause is a disjunction of at most three variables or their negations) over $n$ variables, and tasked with determining if $\varphi$ is satisfiable (i.e, some assignment to the variables makes every clause in $\varphi$ evaluate to true).
The \textsf{Exponential Time Hypothesis} posits that there exists a constant $\delta > 0$ such that \textsf{3SAT} cannot be solved in $2^{\delta n}\poly(n)$ time.

In the \textsf{Set Cover} problem, we are given a set $U$ of $n$ elements, a family $\mathcal{F}$ of subsets of $U$, and a target integer $t$.
We are tasked with determining if there exists a collection of at most $t$ sets from $\mathcal{F}$ whose union equals $U$.
The \textsf{Set Cover Conjecture} asserts that for any constant $\delta > 0$, there exists a positive integer $\Delta$, such that \textsf{Set Cover} on instances where every set in the family $\mathcal{F}$ has size at most $\Delta$ cannot be solved in $2^{(1-\delta)n}\poly(n)$ time.

The \textsf{Exponential Time Hypothesis (ETH)} 
combined with the classical \textsf{NP}-hardness reduction for Hamiltonian Path implies that \textsf{$k$-Path} cannot be solved in $2^{o(k)}\poly(n)$ time \cite{ImpPatZan2001}.
So, it is reasonable to look for algorithms of the form $c^k\poly(n)$ time for \textsf{$k$-Path} and its variants.
However, \textsf{ETH} does not rule out  the possibility that \textsf{$k$-Path} could be solved in $1.001^k\poly(n)$ time (for example). 

In general, no strong lower bounds on the exact time complexity of \textsf{$k$-Path} are known.
However, interesting lower bounds have been proven for stronger versions of this problem.

In the \textsf{$k$-HyperPath} problem, we are given an $r$-graph,\footnote{For a positive integer $r$, recall that an \emph{$r$-graph}, also known as an \emph{$r$-uniform hypergraph} is set of vertices together with subsets of $r$ vertices known as hyperedges. 
A $2$-graph is a graph in the usual sense.} and are tasked with determining if
it contains a sequence of $k$ vertices, such that any $r$ consecutive vertices in the sequence belong to a hyperedge. Assuming the \textsf{Set Cover Conjecture (SeCoCo)}, this generalization of \textsf{$k$-Path} requires $2^{((r-2)/(r-1))k - o(k)}\poly(n)$ time to solve \cite{rkHyperPath}. 
In the \textsf{$k$-Maximum Colored Path} problem, we are given a graph whose vertices are colored, and are tasked with finding a path which passes through at least $k$ vertices with distinct colors. 
Assuming \textsf{SeCoCo}, this generalization of \textsf{$k$-Longest Path} requires  $2^{k - o(k)}\poly(n)$ time to solve \cite[Proposition 2]{shortest-cycle-sequel}.

\section{Bipartitioned Path Algorithm}
\label{sec:near-bipartite}

In the \textsf{$(\ell, k_1, \ell_2)$-Bipartitioned Path} problem, we are given a graph $G$ on $n$ nodes, whose vertices are partitioned into two parts $V_1$ and $V_2$, with distinguished vertices $s$ and $t$, and are tasked with determining if $G$ contains a simple path from $s$ to $t$ of length $\ell$, which uses exactly $k_1$ vertices from $V_1$, and exactly $\ell_2$ edges whose endpoints are both in $V_2$.
Below, we include a proof that \textsf{$(\ell, k_1, \ell_2)$-Bipartitioned Path} can be solved over undirected graphs in $2^{k_1+\ell_2}\poly(n)$ time, following the exposition from \cite[Section 10.4]{cool-parameterized-book}.

The idea behind this algorithm is to construct a polynomial whose monomials correspond to walks of length $k$ from $s$ to $t$ in $G$, some of whose vertices and edges are annotated with labels.
These labels are used in a clever way to sieve out simple paths from this set of walks.

\bipartite*

\begin{proof}
    A walk is a sequence of vertices $u_1, \dots, u_{a+1}$ and edges $e_1, \dots, e_a$ in $G$, such that for each index $i$, $e_i$ is an edge between $u_i$ and $u_{i+1}$.
    Given a set $L$, an $L$-labeled walk is a walk, together with a subsequence of its vertices and edges $f_1, \dots, f_b$, annotated with corresponding elements $c(f_1), \dots, c(f_q)$ of $L$
    (formally, each $f_j$ can be thought of as $(u_i, i)$ or $(e_i, i)$ for some index $i$).
    A vertex or edge $f_j$ annotated with an member of $L$ is called a \emph{labeled element}, and the member $c(f_j)$ of $L$ associated with the labeled element is referred to as its \emph{label}.
    
    Looking ahead, we will introduce a special class of labeled walks, and then construct a polynomial summing over these labeled walks.
    We will then argue that this polynomial is nonzero if and only if there is a solution to the \textsf{$(\ell, k_1, \ell_2)$-Bipartitioned Path} problem.

    Define $\c{W}_{L}$ to be the set of all $L$-labeled walks $W$ with the following properties:

    \begin{enumerate}
        \item The walk $W$ has length $\ell$.
    
        \item The walk $W$ begins at vertex $u_1 = s$ and ends at vertex $u_{\ell+1} = t$.
        
        \item
        For any index $i$, if $u_i\in V_2$ and $u_{i+1} \in V_1$, then $u_{i+2}\neq u_i$.
        In other words, if the walk $W$ leaves a vertex in $V_2$ for a vertex in $V_1$, it never ``immediately returns'' to the same position in $V_2$.

        \item
        The walk $W$ has exactly $k_1$ vertices in $V_1$, and $\ell_2$ edges with both endpoints in $V_2$.

        \item
        The walk $W$ contains exactly $|L|$ labeled elements, and all labels in the walk are distinct.
        Any labeled element in  $W$ is either a vertex in $V_1$ or edge with both endpoints in $V_2$.
    \end{enumerate}

    Given a walk $W\in \c{W}_{L}$, we introduce the corresponding monomial 
        \[f(W) = \prod_{i=1} ^{\ell} x_{u_i, u_{i+1}} \prod_{j=1}^{|L|} y_{f_j, c(f_j)}.\]
     Intuitively, the monomial $f(W)$ uses the $x_{u_i, u_{i+1}}$ variables to keep track of the edges in $W$ and the $y_{f_j, c(f_j)}$ variables to keep track of its labels and labeled elements.

     Now, set $L = [k_1 + \ell_2]$.
     Define the polynomial
        \[P(\vec{x},\vec{y}) = \sum_{W\in\c{W}_{L}} f(W).\]
    A simple dynamic programming argument shows that we can evaluate $P$ at a point in a given field in $2^{k_1+\ell_2}\poly(n)$ field operations.
    The dynamic programming table simply holds evaluations of polynomials which are sums of $f(W)$ for labeled walks $W$ which have length at most $\ell$, start at $s$ and end at some arbitrary vertex $v$, satisfy condition \textbf{3} above, have at most $k_1$ vertices in $V_1$, at most $\ell_2$ edges with both endpoints in $V_2$, and use pairwise distinct labels from $L$.
    The $2^{k_1+\ell_2}$ factor comes from having to perform dynamic programming over all distinct subsets of $L$, and the $\poly(n)$ factor follows from keeping track of counters for the walk length, number of vertices in $V_1$, number of edges with both endpoints in $V_2$, and a constant number of vertices near the ends of the walk $W$.

    Now, take $q = \log (n(\ell + |L|))$ and work over the field $\mathbb{F}_{2^q}$.
    Arithmetic operations over this field can be performed in $\tilde{O}(q)$ time.
    Consider the polynomial $P$ over this field.
    We claim that the net contribution of the monomials of terms in $W\in \c{W}_L$ which are not simple paths (i.e., $W$ has a repeat vertex) vanishes in $P$ over this field.

    Indeed, let $W\in \c{W}_L$ have a repeat vertex.
    There are two cases to consider.

    \paragraph*{Case 1: Repeat in $\boldsymbol{V_1}$}

    Suppose first that $W$ has a repeat vertex in $V_1$.
    Let $i$ be the smallest index such that $u_i\in V_1$ occurs at least twice in $W$. 
    Let $j$ be the smallest index such that $u_j = u_i$.

    By conditions \textbf{4} and \textbf{5}, and the fact that $|L| = k_1 + \ell_2$, we see that $u_i$ and $u_j$ are both labeled elements of $W$, with distinct labels $c_i$ and $c_j$ respectively.

    So, consider the labeled walk $W'$ which has the same vertex and edge sequence, and has the same label sequence as $W$ \emph{except} that $u_i$ has label $c_j$ and $u_j$ has label $c_i$.
    By inspection, $W'\in \c{W}_L$ and $f(W) = f(W')$.
    Moreover, applying this label swap argument to $W'$ recovers $W$.
    So this pairs up the walks $W$ and $W'$ together, and over $\mathbb{F}_{2^q}$ the contribution $f(W) + f(W')$ will vanish.

    \paragraph*{Case 2: All Repeats in $\boldsymbol{V_2}$}

    Suppose $W$ contains no repeats in $V_1$.
    Let $i$ be the smallest index such that $u_i$ occurs at least twice in $W$. 
    By assumption, $u_i\in V_2$.
    Let $j$ be the smallest index such that $u_j = u_i$. 
    In this case, consider the closed walk $C = W[u_i, u_j]$ between these first two occurrences of $u$ in $W$.

    If the sequence of vertices in $C$ is not a palindrome, we can construct the labeled walk $W'$ formed by reversing $C$ (both the vertices and corresponding labels) in $W$. 
    By inspection, $W'\in \c{W}_L$ and $f(W') = f(W)$.
    Moreover, applying this reversing procedure on $W'$ recovers $W$.
    So this pairs up the walks $W$ and $W'$ together, and over $\mathbb{F}_{2^q}$ the contribution $f(W) + f(W')$ will vanish.
    
    Otherwise, suppose that the vertices of $C$ do form a palindrome.
    Then the sequence of vertices for $C$ looks like
        \[u_i, u_{i+1}, \dots, u_{(i+j)/2}, \dots, u_{j-1}, u_j.\]
    where $u_{i+p} = u_{j-p}$ for all $p\le (j-i)/2$.
    
    Since $W$ has no repeat vertices in $V_1$, we know that 
    $u_{i+p}, u_{j-p}\in V_2$ for all $p < (j-i)/2$.
    This together with condition \textbf{3} implies that  $u_{(i+j)/2}\in V_2$.
    Thus, $C$ has 
    at least two edges with both endpoints in $V_2$.
    By conditions \textbf{4} and \textbf{5} together with the fact that $|L| = k_1 + \ell_2$, we see that all such edges are labeled in $W$.
    So, take the edges $e_i = (u_i, u_{i+1})$ and $e_{j-1} = (u_{j-1}, u_j)$ with distinct labels $c_i$ and $c_{j-1}$.
    Since $C$ is a palindrome, $e_i = e_{j-1}$.

    Thus, we can consider the labeled walk $W'$ which has the same vertex and edge sequence, and has the label sequence as $W$ \emph{except} that $e_i$ has label $c_j$ and $e_j$ has label $c_i$.
     By inspection, $W'\in \c{W}_L$ and $f(W) = f(W')$.
    Moreover, applying this label swap argument to $W'$ recovers $W$.
    So this pairs up the walks $W$ and $W'$ together, and over $\mathbb{F}_{2^q}$ the contribution $f(W) + f(W')$ will vanish.

Thus, in either case, we see that all elements of $\c{W}_L$ with repeat vertices can be paired up so that their monomials have net zero contribution over $\mathbb{F}_{2^q}$ in the definition of $P$.
On the other hand, by definition, all $W\in\c{W}_L$ with no repeat vertices (i.e., $W$ whose vertex and edge sequences correspond to simple paths) produce distinct $f(W)$.
Thus, we see that over $\mathbb{F}_{2^q}$, the polynomial $P$ is nonzero if and only if there is a solution to the 
\textsf{$(\ell, k_1, \ell_2)$-Bipartitioned Path} problem.

Take a uniform random evaluation of $P$ over $\mathbb{F}_{2^q}$.
By the Schwartz-Zippel Lemma, this evaluation is nonzero with probability at least $1 - (\ell + |L|)/2^q = 1 - 1/n$ if $P$ is nonzero (and of course is always zero if $P$ is the zero polynomial).
Thus, we can solve the problem by checking whether a uniform random evaluation of $P$ over $\mathbb{F}_{2^q}$ is nonzero, which by the preceding discussion can be done in $2^{k_1+\ell_2}\poly(n)$ time.
\end{proof}

\end{document}